\documentclass[a4paper,UKenglish]{lipics}
 
\title{Graph-theoretical Bounds on the Entangled Value of Non-local Games}

\author[1]{Andr\'e Chailloux}
\author[2]{Laura Man\v{c}inska}
\author[3]{Giannicola Scarpa}
\author[4]{Simone Severini}
\affil[1]{SECRET Project Team, INRIA Paris-Rocquencourt,\\
   Paris, France\\
  \texttt{andre.chailloux@inria.fr}}
\affil[2]{Center for Quantum Technologies,\\
   Singapore\\
  \texttt{cqtlm@nus.edu.sg}}
\affil[3]{Universitat Aut\`onoma de Barcelona,\\
   Barcelona, Spain\\
  \texttt{giannicola.scarpa@uab.cat}} 
\affil[4]{University College London,\\
   London, United Kingdom\\
  \texttt{simoseve@gmail.com}}
\authorrunning{A.\,Chailloux, L.\,Man\v{c}inska, G.\,Scarpa and S.\,Severini} 

\Copyright{Andr\'e Chailloux, Laura Man\v{c}inska, Giannicola Scarpa and Simone Severini}

\subjclass{G.2.3 Applications}
\keywords{Graph theory, non-locality, entangled games}

\serieslogo{}
\volumeinfo
  {}
  {0}
  {9\textsuperscript{th} Conference on the Theory of Quantum Computation, Communication, and Cryptography}
  {1}
  {1}
  {1}
\EventShortName{}
\DOI{10.4230/LIPIcs.xxx.yyy.p}

\def\01{\{0,1\}}

\newcommand{\ket}[1]{|#1\rangle}
\newcommand{\bra}[1]{\langle#1|}

\newcommand{\inp}[2]{\langle{#1}|{#2}\rangle} 

\newcommand{\Tr}{\mbox{\rm Tr}}

\newcommand{\C}{\mathbb{C}}

\newcommand{\supp}{\operatorname{supp}}

\newcommand{\ie}{\emph{i.e.}}

\begin{document}
\maketitle

\begin{abstract}
We introduce a novel technique to give bounds to the entangled value of non-local games. The technique is based on a class of graphs used by Cabello, Severini and Winter in 2010. The upper bound uses the famous Lov\'asz theta number and is efficiently computable; the lower one is based on the quantum independence number, which is a quantity used in the study of entanglement-assisted channel capacities and graph homomorphism games.
\end{abstract}  

\section{Introduction}

In non-local games, two non-communicating players cooperate in order to achieve a task. Each player receives an input and produces an output, and they must satisfy the task's requirements. 

In physics, this class of games is also known as ``entangled games''. They are mostly used to investigate the power of entanglement, by designing intuitive Bell inequalities. One designs a non-local game and proves an upper bound on the winning probability of the classical players (the Bell inequality). Later, one shows that there exists a quantum strategy that by using entanglement can beat that winning probability. Two famous examples of such approach are the CHSH game (based on \cite{chsh}) and the magic square game (based on \cite{peres}).

Non-local games are also important  in computer science, where they are usually called ``two-prover one-round games''. Their intuitive nature has been used in complexity theory to approach the difficult problem of P vs. NP, by defining probabilistically checkable proofs and ultimately leading to the famous unique games conjecture \cite{Khot02,Khot10}. 

Estimating or bounding the value of a game given its description is an important task, and much effort has been devoted to the question. For example, the entangled value for the class of XOR games has been shown to be easy to compute with a semidefinite program by Cirel'son \cite{Tsirelson80}. Also the entangled value of unique games turns out to be easy to compute, therefore falsifying the unique games conjecture in the quantum world \cite{KRT08}.

Here, we propose a general approach to bound the value of a non-local game based on graph theory. 
Given the description of a game, we construct a graph that contains all the information about the game, and we call it the ``game graph''. The construction is based on the techniques in \cite{CSW10}.
Such techniques have also been extended and used in \cite{acin}.

We first show that the classical value of any game is equal to the independence number of its game graph (renormalized). This reflects the fact that computing exactly the classical value of a game is NP-hard.
We then show an efficiently computable upper bound on the quantum value of a game (and therefore on the classical value), given by the celebrated Lov\'asz theta number. 
We then give lower bound for the games on the uniform distribution given by the quantum independence number, a graph parameter introduced in \cite{CLMW09} and futher discussed in \cite{MSS13,Roberson:2012}. 
To conclude, we give  a class of games for which this upper bound is tight.

We believe this graph-theoretical approach is an important and a fertile field for improvements. We discuss these in the conclusions section.

\section{Preliminaries} \label{sec:defch3}

\subsection{Non-local games}

We now briefly describe the setting of a non-local game $\mathcal G$.

Alice and Bob are separated and forbidden to communicate. They receive inputs $x$ and $y$ from some input sets $X$ and $Y$, according
to some fixed and known probability distribution~$\pi$, and are required to produce outputs $a$ and $b$ from output sets $A$ and $B$, respectively.
The game rules are encoded in a predicate $\lambda: X\times Y\times A\times B \rightarrow \{0,1\}$, which specifies which outputs $a,b$ are correct on inputs $x,y$. In other words, players win the game on inputs $x,y$ if they output some $a,b$ such that $\lambda(x,y,a,b)=1$. 
The goal of the players is to maximize the winning probability.

A \emph{classical strategy} for the game is without loss of generality a pair of functions, $f_A : X \to A$ for Alice and $f_B: Y \to B$ for Bob. (Shared randomness between the two players is easily seen not to be beneficial.)
The winning probability of a strategy is calculated as follows: 
\[ \sum_{x,y} \pi(x,y) \lambda(x,y,f_A(x),f_B(y)).\]
The \emph{classical value} $\omega(\mathcal{G})$ of the game is the maximum winning probability among all classical strategies.

In \emph{entangled strategies} (a.k.a.\ quantum strategies), players share a fixed (\ie, independent of the inputs) entangled state $\ket{\psi}$.
For each input $x$, Alice has a projective measurement  $\{P^x_a\}_{a\in A}$, and for each input $y$, Bob has a projective measurement $\{Q^y_b\}_{b\in B}$. Upon receiving the inputs, they apply the corresponding measurements to their parts of the entangled state and produce classical outputs $a$ and $b$, respectively.
The winning probability of a strategy is calculated as follows: 
\[
\sum_{xy} \pi(x,y) \lambda(x,y,a,b) \langle\psi|P^x_a\otimes Q^y_b|\psi\rangle.
\]

The \emph{entangled value} $\omega^*(\mathcal{G})$ of the game is the supremum of the winning probability,
taken over all entangled strategies.

A Bell inequality for a game is an upper bound on its classical value. We have a Bell inequality violation for a game $\mathcal G$ if the entangled value is strictly larger than the classical one. The violation is quantified by the ratio $\omega^*(\mathcal{G}) / \omega(\mathcal{G})$. 

The CHSH game is one particularly famous example~\cite{chsh}.
Here, the inputs $x\in\01$ and $y\in\01$ are uniformly
distributed, and Alice and Bob win the game if their respective outputs $a\in\01$ and $b\in\01$
satisfy $a\oplus b=x\wedge y$; in other words, $a$ should equal $b$ unless $x=y=1$.
The classical value of this game is easily seen to be $\omega(\mathcal{G})=3/4$, while the
entangled\index{entanglement} value is known to be $\omega^*(\mathcal{G})=1/2+1/(2\sqrt{2})\approx 0.85$.

A non-local game is said to be a \emph{pseudo-telepathy} game if the quantum value is 1 while the classical value is strictly less than 1.

\subsection{Notions of graph theory} \label{sec:graphs_intro}

A \emph{simple graph} $G=(V,E)$ consists of a finite vertex set $V$ and its edge set $E \subsetneq V\times V$ (the inclusion here is strict because there are no edges of the form $(v,v)$). Two vertices $(v,w) \in E$ are ``adjacent'' or equivalently ``form an edge''.
All graphs considered here are simple graphs. 
For a graph $G=(V,E)$, we also denote its vertex set with $V(G)$ and its edge set with $E(G)$ whenever confusion has to be avoided. 

An \emph{independent set} of a graph is a subset $I$ of $V(G)$  such that no two elements of $I$ are adjacent.
The \emph{independence number} of a graph $G$, denoted by $\alpha(G)$, is the maximum size of an independent set of $G$.


%

A $d$-dimensional orthogonal representation\index{orthogonal representation} of $G=(V,E)$ is a map $\phi: V \rightarrow \C^d$ 
such that for all $(v,w)\in E$, $ \inp{\phi(v)}{\phi(w)} = 0$.  (If all the vectors have unit norm, this is called orthonormal representation.)

We finally introduce an important graph parameter: the \emph{theta number} (a.k.a.\ Lov\'{a}sz number or theta function). \index{Lov\'{a}sz!$\vartheta $ number} \index{Lov\'{a}sz}
It  was originally defined by Lov\'{a}sz~\cite{Lovasz} to solve a long-standing problem posed by Shannon~\cite{Shannon:1956}: computing the Shannon capacity\index{Shannon!capacity} of the five-cycle.
There are many equivalent formulations of the theta number (see \cite{Knuth:1993} for a detailed survey).
The one that we use in this paper is the following:
\begin{equation} \label{eq:theta_def1}
\vartheta(G)=\max\sum_{v\in V(G)} |\langle\psi|\psi_{v}\rangle|^{2},
\end{equation}
where the maximum is over
unit vectors $\psi$ and orthonormal representations $\{\psi_{v}\}_{v\in V(G)}$.
Lov\'asz~\cite{Lovasz} proved that $\alpha(G)\leq \vartheta(G)$ holds (this inequality is part of the so-called ``sandwich theorem''~\cite{Knuth:1993}).
The theta number can be approximated to within arbitrary precision in polynomial time, hence it gives a tractable and in many cases useful bound for $\alpha$.

\subsection{Quantum Independence Number}

In this section we define the quantum independence number\index{independence number!quantum}  and state some of its properties.

First, let us briefly give some historical background. In~\cite{MSS13} the concept of quantum independence number is presented in the context of zero-error information theory. This quantity is usually called in literature ``one-shot zero-error entanglement-assisted channel capacity'' and denoted as $\alpha^{*}$. 
A new definition of quantum independence number, denoted as $\alpha_q$, came in~\cite{Roberson:2012}, in the context of graph homomorphisms. As of today, it is not known if the two quantities are equal for all graphs. 
In this paper we use the second quantity, but for simplicity we omit the details about homomorphisms and provide a direct definition.

As with the quantum chromatic number (see \cite{qchrom}), the quantum independence number can be defined in terms of a non-local game. 
Informally, the \emph{independent set game}\index{independent set game} with parameter $t$ for a graph $G=(V,E)$ is as follows.
Two players, Alice and Bob, claim that they know an independent set $I$ of $G$ consisting of $t$ vertices. 
A referee wants to test this claim with a non-local game. He forbids communication between the players, generates two uniformly random numbers $x,y \in [t]$
and separately asks Alice to provide the $x$-th vertex of $I$ and Bob to provide the $y$-th vertex of $I$. 
The players are required to output the same vertex if $x=y$, and to output non-adjacent vertices if $x\neq y$. A formal definition follows.

\begin{definition}\index{non-local game!independent set}
The \textit{independent set game with parameter $t$} on the graph $G=(V,E)$ is a non-local game with input sets $X=Y=[t]$, output sets $A=B=V$. The probability distribution $\pi$ is the uniform distribution on the input pairs. Alice gets input $x$ and outputs $v$, Bob gets input $y$ and outputs $w$. The players \emph{lose the game} in the following two cases:
\begin{enumerate}
\item $x=y$ and $v\neq w$
\item $x\neq y$ and $(v,w)\in E$ or $v=w$
\end{enumerate}
\end{definition}

A classical strategy consists w.l.o.g.\ of two deterministic functions $f_A: [t] \rightarrow V$ for Alice and 
\mbox{$f_B: [t] \rightarrow V$} for Bob. Shared randomness, as seen for the coloring game, is not beneficial.
A little thought will show that to win with probability 1, we must have $f_A = f_B$ 
(to avoid the first losing condition) and that $\{f_A(1),\dots, f_A(t)\} $ must be a valid independent set of the graph of size $t$ (to avoid the second losing condition). 
It follows that the classical players cannot win the game with probability $1$ when $t>\alpha(G)$. 

It is proven in~\cite{Roberson:2012} that w.l.o.g.\  quantum strategies for the independent set game\index{non-local game!independent set}
consist of projective measurements on a maximally entangled\index{entanglement!maximally entangled} state, that the projective measurements of Alice and Bob are the same and that all the projectors can be real-valued. Therefore we can define a \emph{quantum independent set} of size $t$ as a collection of $t$ projective measurements $\{P_{v}^{x}\}_{v\in V}$ for all $x\in [t]$ that have the whole vertex set as outputs, 
with the following consistency condition:
\begin{equation} \label{eq:consistency_indep}
\mbox{for all } (u,v) \in E \mbox{ or } u=v \mbox{ and for all } x\neq x', \quad P_{u}^{x} P_{v}^{x'}=0.
\end{equation}

\begin{definition} For all graphs $G$, the \emph{quantum independence number}\index{independence number!quantum}  $\alpha_q(G)$
 is the maximum number $t$ such that there exists a quantum independent set of $G$ of size $t$.
\end{definition}

\section{Game graphs}

\subsection{Definition and relation to $\omega(\mathcal{G})$}

Consider a non-local game $\mathcal{G}$ with input sets $X,Y$, output sets $A,B$, predicate $\lambda: X\times Y\times A\times B \rightarrow \{0,1\}$ and uniform distribution on the inputs.

\begin{definition} \label{def:gamegraph} \index{non-local game!graph associated to}
A graph $G=\left(V,E\right)$ \emph{associated} to the game $\mathcal{G}$ has:
\begin{enumerate}
\item $V=\{xyab \mid x\in X,y\in Y,a\in A,b\in B \mbox{ and }\lambda(x,y,a,b)=1 \}$,
\item $E=\{\{xyab,x^{\prime}y^{\prime}a^{\prime}b^{\prime}\}\mid (x=x^{\prime}\wedge a\neq a^{\prime})\vee(y=y^{\prime}\wedge b\neq b^{\prime})\}$.
\end{enumerate}
\end{definition}

This definition is inspired by a construction in~\cite{CSW10} in the framework of contextuality of physical theories. The authors used something similar to Definition \ref{def:gamegraph} for the special case of the CHSH game.\index{non-local game!CHSH} Here we generalize to all games.

For simplicity, we prove the results in this section for the case where the game has the uniform distribution on the inputs and $\lambda$ is a boolean function. It is straightforward to generalize to games with real-valued predicate and any probability distribution $\pi$ of the inputs, as follows. 
Consider the (vertex) weighted graph with all the quadruples $xyab$ in the vertex set, labelled with $\mbox{weight}(xyab)=\lambda(x,y,a,b)\cdot\pi(x,y)$, and the same edge set as before.
The classical bound and the Lov\'asz theta bound that we will prove later can be adapted by considering the weighted versions of these parameters. However, we do not know how to generalize our last result because we do not define the quantum independence number for a weighted graph. 
\index{Lov\'{a}sz!$\vartheta $ number}

Now we prove that that the classical value of a game can be expressed in terms of the independence number of its game graph.\index{independence number}
\begin{theorem}
\label{theorem:omega_c}Let $\mathcal{G}$ be a non-local game with input sets $X$ and $Y$,  uniform input distribution and associated
graph $G$. Then 
\[
\omega(\mathcal{G})= \frac{\alpha(G)} {|X\times Y|}.
\]
\end{theorem}
\begin{proof}
Let $k=|X\times Y|.$ We begin by proving
that $\omega(\mathcal{G})\geq\alpha(G)/k$. Namely, we show that
given a maximal independent set $I\subseteq V$ of size $\ell$,
we can exhibit a strategy for $\mathcal{G}$ that answers correctly
to at least $\ell$ of the $k$ questions. By the structure of $G$,
the independent set~$I$ cannot contain vertices $xyab$ and $xy^{\prime}a^{\prime}b^{\prime}$
such that $a\neq a^{\prime}$. Similarly, $I$ cannot contain vertices
$xyab$ and $x^{\prime}ya^{\prime}b^{\prime}$ such that $b\neq b^{\prime}$.
Hence, we have the following strategy: on input $x$, Alice outputs
the unique $a$ determined by the vertices in the independent set~$I$. Bob behaves similarly. Since $V$ contains only winning quadruples $xyab$, the
size $\ell$ of the independent set means Alice and Bob answer correctly
to at least $\ell$ input pairs. Hence, $\omega(\mathcal{G})\geq \ell/k$. 

Now we show that $\omega(\mathcal{G})\leq\alpha(G)/k$, \emph{i.e.},
if there exists a strategy that wins on $\ell$ of the
$k$ input pairs, then there exists an independent set with weight
$\ell$. 
We have that w.l.o.g.~classical strategies consist of a pair of functions. Fix Alice and Bob's  functions $f_{A}$ and $f_{B}$ that win on $\ell$ input pairs. Now take the set of quadruples $S=\{(x,y,f_{A}(x),f_{B}(y))\}_{x\in X, y\in Y}$.
We have that $I = S\cap V$ is a set of $\ell$ vertices of $G$.  
Since $f_A$ and $f_B$ are deterministic, $I$ cannot contain vertices
$xyab$ and $xy^{\prime}a^{\prime}b^{\prime}$
such that $a\neq a^{\prime}$ 
nor vertices $xyab$ and $x^{\prime}ya^{\prime}b^{\prime}$ such that $b\neq b^{\prime}$. Therefore, there cannot be an edge between any pair
of the elements of $I$ and we have that $I$ is an independent set of $G$ of size $\ell$. 
Hence, $\alpha(G) \geq \ell$.
Combining the two directions proves the theorem. 
\end{proof}

\subsection{Bounds on the entangled value of a game}

Cabello, Severini and Winter \cite{CSW10}\index{Cabello} observe that the quantum value of the CHSH game\index{non-local game!CHSH} is equal to the theta number of its associated graph divided by the number of questions.
We have found by direct calculation that this is not always true for general games, for example in the case of the 2-fold parallel repetition of CHSH.
The same conclusion follows from the results of Ac\'{i}n \emph{et al.} in \cite{acin}.
Here we prove the upper bound directly for our specific constructions.

\begin{theorem}\label{theorem:omega_q_lt_theta}
Let $\mathcal{G}$ be a non-local game with input sets $X$ and $Y$, uniform input distribution and associated
graph $G=(V,E)$. Then 
\[\omega^*(\mathcal{G})\leq\frac{\vartheta(G)}{|X\times Y|}.\]
\end{theorem} 
\begin{proof}
Let $k = |X\times Y|$.
Consider a quantum strategy for $\mathcal{G}$ that achieves the value
$\omega^*(\mathcal{G})$. It consists of a shared entangled\index{entanglement} state
$|\psi\rangle$ and a collection of projective measurements $\{P^x_a\},\{Q^y_b\}$,
such that 
\[
\sum_{xyab}\frac{1}{k}\lambda(x,y,a,b)\langle\psi|P^x_a\otimes Q^y_b|\psi\rangle=
\frac{1}{k} \sum_{xyab \in V} \langle\psi|P^x_a\otimes Q^y_b|\psi\rangle=
\omega^*(\mathcal{G}).
\]

For each quadruple $xyab$ let $|\psi_{xyab}\rangle=P^x_a\otimes Q^y_b|\psi\rangle.$
This is an orthogonal representation of $G$, since for every edge $(xyab,x'y'a'b')$
either $P^x_aP^{x'}_{a'}=0$ or $Q^y_bQ^{y'}_{b'}=0$. 
Now for each $xyab$ consider the normalized vector  \index{orthogonal representation}
\[ 
|\psi'_{xyab}\rangle = \frac{|\psi_{xyab}\rangle}{||\psi_{xyab}||} = \
\frac{|\psi_{xyab}\rangle}{\sqrt{\langle\psi|P^x_a\otimes Q^y_b|\psi\rangle}}.
\]

We have that $\{\psi'_{xyab}\}_{xyab\in V}$
and $\psi$ are a feasible solution for the formulation~\eqref{eq:theta_def1} of $\vartheta(G)$.

We conclude 
\begin{align*}
\vartheta(G) & \geq \sum_{xyab\in V} |\langle\psi|\psi_{xyab}\rangle|^2  \\
& = \sum_{xyab \in V} \left| \frac{\langle\psi|\psi_{xyab}\rangle}{||\psi_{xyab}||} \right|^2 \\
& = \sum_{xyab \in V} \frac{\langle\psi|P^x_a\otimes Q^y_b|\psi\rangle^2}{\langle\psi|P^x_a\otimes Q^y_b|\psi\rangle} \\
& = \sum_{xyab \in V} \langle\psi|P^x_a\otimes Q^y_b|\psi\rangle \\
& = k \cdot \omega^*(\mathcal{G}).
\end{align*}
\end{proof}

We now have the following lower bound in terms of the quantum independence~number.\index{independence number!quantum}
\begin{theorem}
\label{theorem:alpha_q_lt_omega_q}
Let $\mathcal{G}$ be a non-local game with input sets $X$ and $Y$, uniform input distribution and associated graph $G=(V,E)$. Then 
\[
\omega^*(\mathcal{G}) \geq \frac{\alpha_{q}(G)}{|X\times Y|}
\]

\end{theorem}
To prove the theorem, we will use the following lemma.
\begin{lemma}
Let $M,N$ be positive semidefinite matrices. Then for any vector $\ket{v}$, we have that
\[\bra{v}\supp(M+N)\ket{v}\geq \bra{v}\supp(M)\ket{v},\]
where $\supp(M)$ denotes the projector onto the support (\emph{i.e.,}~the column space) of~$M$.
\label{lem:LinAlg}
\end{lemma}
\begin{proof}
If $P$ is a projector onto a subspace $\Pi$ then $\bra{v}P\ket{v}$ is the squared length of the projection of $\ket{v}$ into $\Pi$. Hence, to prove the lemma it suffices to show that $\supp(M)\subseteq \supp(M+N)$, where by abusing the notation we use $\supp$ to denote the support itself (rather than the projection onto it). 

For contradiction, suppose that $\supp(M)\not\subseteq\supp(M+N)$. Then the orthogonal complement of $\supp(M)$ (i.e.~the nullspace $\operatorname{Null}(M))$ does not contain $\operatorname{Null}(M+N)$. Hence we can pick a vector $\ket{w}$ such that $(M+N)\ket{w} = 0$ but $M\ket{w} \neq 0$. This further implies that 
\[
  \bra{w} N \ket{w}  = \bra{w} (M+N) \ket{w} - \bra{w} M \ket{w} 
  = - \bra{w} M \ket{w} <0,
\]
 since $M$ is positive semidefinite and $M\ket{w}\neq 0$. This completes the proof as we have reached a contradiction with the initial assumption that $N$ is positive semidefinite.
\end{proof}

\begin{proof}[Proof of Theorem~\ref{theorem:alpha_q_lt_omega_q}.]
Given a quantum strategy $\{P_{xyab}^{i}\}$ for the independent set game\index{non-local game!independent set} on $G$ with parameter $t$, we construct a strategy to win the game $\mathcal{G}$ with probability at least $t/|X\times Y|$, as follows.

Players share a maximally entangled\index{entanglement!maximally entangled} state with local dimension $d$
(which is the dimension of the projectors above). On input $x,$ Alice
measures her half of the state using the projective measurement$\{P^x_a\}_{a\in A}\bigcup\{I-\sum_{a}P^x_a\}$,
where the individual elements are defined as follows: 
\[
  P^x_a=
  \supp\left( \sum_{\stackrel{yb}{xayb\in V}} \sum_i
  P_{xayb}^{i}\right),
\]
where we use $\supp(M)$ to denote the projector onto the image of $M$. 
We show that this is a valid projective measurement. For all $y,b,y',b'$ there is an edge $(xyab,xy'a'b')\in E$. Therefore in the strategy for the independent set game\index{non-local game!independent set} we have that for all $i,j$ each projector $P^i_{xyab}$ is orthogonal to $P^j_{xy'a'b'}$. 
Hence, for all $a\neq a'$ we have $P^x_a\cdot P^x_{a'}=0$. 
Bob constructs projectors $P^y_b$ similarly. 

Now we lower bound the quantum value of $\mathcal{G}$ as follows:
\begin{eqnarray*}
 &|X\times Y|\cdot\omega^*(\mathcal{G}) & \geq  \sum_{xyab\in V}\langle\psi|P^x_a\otimes P^y_b|\psi\rangle\\
 &  & =\sum_{xyab \in V}\langle\psi|
  \supp \Big(\sum_{i,j} \sum_{\stackrel{y'b'}{xay'b'\in V}}
  \sum_{\stackrel{x'a'}{x'a'yb\in V}}
  P_{xay'b'}^{i}\otimes P_{x'a'yb}^{j}\Big)
  |\psi\rangle,
\end{eqnarray*}
where we have used the fact that $\supp(M \otimes N) = \supp(M)\otimes\supp(N)$ for all matrices $M,N$ to obtain the last equality. Now by applying Lemma~\ref{lem:LinAlg}, we drop all the terms except the ones with $i=j,a=a',b=b',x=x'$ and $y=y'$, \mbox{and we have that}
\begin{align}
  |X\times Y|\cdot\omega^*(\mathcal{G}) & \geq
  \sum_{xyab \in V}\langle\psi|
  \supp \Big(\sum_{i} P_{xayb}^{i}\otimes P_{xayb}^{i}\Big) |\psi\rangle\\
  & = \sum_{xyab \in V}\langle\psi|
   \Big(\sum_{i} P_{xayb}^{i}\otimes P_{xayb}^{i}\Big) |\psi\rangle
  \label{eq:Projs}\\
  & = \sum_{xyab \in V} 
     \sum_i \frac{1}{d} \Tr(P_{xayb}^i)\label{eq:MaxEnt}\\
  & = \sum_i \frac{1}{d} \Tr(I_d)\label{eq:Measurement}\\
  & = \alpha_q(G).
\end{align}
In the above we have observed that $\supp(P+Q)=P+Q$ for mutually orthogonal projectors $P$ and $Q$ to get Expression~(\ref{eq:Projs}). We have used properties of $\ket{\psi} = \frac{1}{\sqrt{d}}\sum_i \ket{i,i}$ to obtain Expression~(\ref{eq:MaxEnt}). We have used the fact that, for all $i$, $\{P^i_{xayb} : \lambda(x,a,y,b)=1\}$ forms a measurement to obtain Expression~(\ref{eq:Measurement}).
\end{proof}

\subsubsection{Tightness of the lower bound}

Here we obtain an equality relation between the value of the game and the quantum independence number of the game graph, for a class of pseudo-telepathy games.

\begin{theorem}
\textup{Let $\mathcal{G}$ be a pseudo-telepathy game with a 0-1 valued
predicate~$\lambda$, admitting a quantum strategy consisting of a maximally entangled state $|\psi\rangle$ and pairwise commuting projectors.
Let $G$ be the corresponding game graph. Then, \[\omega^*(\mathcal{G})=\frac{\alpha_{q}(G)}{|X\times Y|}=1.\]}\end{theorem}
\begin{proof}
From Theorem \ref{theorem:alpha_q_lt_omega_q} we have $\alpha_{q}(G)\leq|X\times Y|\cdot\omega^*(\mathcal{G})$.
We need to prove the other direction.

Let $\{P^x_a\},\{Q^y_b\}$ be the strategies that win the game $\mathcal{G}$
on $|\psi\rangle$. We have: 
\[
\sum_{xy}\pi(x,y)\sum_{ab : \lambda(xyab)=1}\langle\psi|P^x_a\otimes Q^y_b|\psi\rangle=1,
\]
so for all $(x,y)$ we must have

\[
\sum_{ab : \lambda(xyab)=1}\langle\psi|P^x_a\otimes Q^y_b|\psi\rangle=1
\]
and for all quadruples $(x,y,a,b)$ such that $\lambda(xyab)=0$ we
have $P^x_aQ^y_b=0$.

Let $\Pi_{xyab}=P^x_aQ^y_b.$ These are projectors thanks to the commutativity assumption. We observe:
\begin{enumerate}
\item For all $(x,y)$ we have 
\[
\sum_{ab : \lambda(xyab)=1}P^x_aQ^y_b=\sum_{\stackrel{ab}{}}P^x_aQ^y_b=\sum_{a}P^x_a\sum_{b}Q^y_b=I,
\]
where the second equality follows from $Q^y_b Q^y_{b'}=\delta_{bb'}$. 
\item For each edge $(x,y,a,b),(x',y',a',b')$ we have a collection of $t$ real-valued projective measurements $\{P_{v}^{x}\}_{v\in V}$ for all $x\in [t]$ that have the whole vertex set as outputs, 
\[
\Pi_{xyab}\Pi_{x'y'a'b'}=0,
\]
because if $x=x'$ and $a\neq a'$ then $P^x_aP^x_{a'}=0$, and if
$y=y'$ and $b\neq b'$ then $Q^y_b Q^y_{b'}=0$.
\end{enumerate}

Therefore, we can construct $|X\times Y|$ projective measurements
that are a winning strategy for the independent set game\index{non-local game!independent set} with \mbox{$t=|X\times Y|$} 
as follows. For each pair $(x,y)$ consider the projective
measurement $\{\Pi_{xyab}\}_{a,b:\lambda(xyab)=1}$ (and zero matrices
for the other vertices of the graph). The first observation above
proves that those are valid projective measurements; the second observation
shows that they respect the consistency condition \eqref{eq:consistency_indep}.
\end{proof}

\section{Concluding remarks and open problems}

We have formalized and discussed a novel approach for the study of non-local game in a combinatorial fashion. Work in progress on this approach relate to the easy generalization to more than 2 players, and the less-easy computation of graphs for the parallel repetition of games.

Our approach has ample room for improvement. Open questions include: 
\begin{enumerate}

\item Can we find a tighter lower bound for the entangled value of all games by using some variant of the quantum independence number, such as the one  in \cite{BBLPS13}?  
Alternatively, can we prove tightness of the current lower bound?

\item Can we find better lower bounds, for example using one of the variants of Lov\'asz theta number?

\item Can we characterize the class of games for which the Lov\'asz bound is tight? We know that the value of CHSH is exactly the theta number of its game graph (see \cite{CSW10}). Is this true for all the XOR games? This would reflect the fact that their value is easy to compute. 

\item Are there other graph parameters related to the classical and entangled values of specific classes of games, for example unique games?

\item We  have shown that for a class of pseudo-telepathy games that quantum players can win using commutative projective measurements on maximally entangled state, this bound is tight. 
A similar class of games is shown in \cite{MSS13} to be in  one-to-one correspondence with a generalization of Kochen-Specker sets. It is not clear to us if those two results together could be used to prove something stronger. 
Perhaps the whole class could be interpreted as pseudo-telepathy games based on some graph parameter (maybe the homomorphism games in \cite{Roberson:2012}) and the relationship to the quantum independence number would be a consequence of this.

\end{enumerate}

\paragraph*{Acknowledgements}
The authors thank the referees of TQC 2014 for useful comments.
Laura Man\v{c}inska is supported by the MOE Tier 3 Grant ``Random numbers from quantum processes'' (MOE2012-T3-1-009).
G.\,Scarpa is supported by the European Commission project RAQUEL(323 970). Part of this work was done while G.\,Scarpa was a PhD student at CWI, supported by Ronald de Wolf's VIDI grant from NWO. S.\,Severini is supported by the Royal Society and EPSRC.


\newcommand{\etalchar}[1]{$^{#1}$}

\end{document}